\documentclass[11pt]{article}
\usepackage{amsmath, graphicx, hyperref, amsthm, cleveref}
\usepackage{booktabs,multirow,tabularx,adjustbox}
\newcommand{\ra}[1]{\renewcommand{\arraystretch}{#1}}
\usepackage[linesnumbered,boxed,noline]{algorithm2e}

\usepackage{isomath}
\usepackage{latexsym, amsfonts,amssymb,amstext}
\usepackage{fullpage} 
\usepackage[usenames]{color}
\usepackage{nicefrac,setspace}
\usepackage{multicol}
\usepackage[framemethod=tikz]{mdframed}
\usepackage{xcolor}
\usepackage{changes}
\usepackage[margin=20pt,
font+=small,labelformat=parens,labelsep=space,
skip=6pt,list=false,hypcap=false
]{subcaption}
\usepackage{wrapfig}
\usepackage{enumitem}
\hypersetup{colorlinks,
	citecolor=blue, 
	linkcolor=black,
	urlcolor=black}

\usepackage{mathpazo} % math & rm
\usepackage{eulervm}
\usepackage[scaled]{helvet} % ss
\usepackage{courier} % tt
\usepackage[T1]{fontenc}
\usepackage[none]{hyphenat}
\usepackage{amsfonts}

\definechangesauthor[name=cyrus, color=purple]{c}
\definechangesauthor[name=siva, color=blue]{s}

\newtheorem{theorem}{Theorem}

\newtheorem{proposition}[theorem]{Proposition}
\newtheorem{lemma}[theorem]{Lemma}
\newtheorem{corollary}[theorem]{Corollary}

\newtheorem*{definition}{Definition}
\newtheorem*{definition*}{Definition}
\theoremstyle{definition}
\newtheorem*{remark}{Remark}
\crefname{equation}{Eq.}{Eqs.}

\newcommand{\Ex}[2]{\mathop{\mathbb{E}}\displaylimits_{#1}\left
[ #2 \right ]}

\newcommand{\floor}[1]{\left\lfloor #1 \right\rfloor}
\newcommand{\ceil}[1]{\left\lceil #1 \right\rceil}

\newcommand{\F}{\mathbb{F}}
\newcommand{\trans}[1]{#1^\intercal}
\newcommand{\m}{\mathsf{mat}}
\newcommand{\ve}{\mathsf{vec}}
\newcommand\inner[2]{\left\langle #1, #2 \right\rangle}
\newcommand\dH[2]{d_H\left(#1,#2 \right)}

\DeclareMathAlphabet{\mathcal}{OMS}{cmsy}{m}{n}

\newcommand{\relmiddle}[1]{\mathrel{}\middle#1\mathrel{}}
\newcommand{\T}{T}
\newcommand{\umv}{\trans{\mathsf{u}}\mathsf{Mv}}
\newcommand\nth[1]{#1^\mathrm{th}}

\title{Equivalence of Systematic Linear Data Structures and \\ Matrix Rigidity}

\begin{document}
\date{\today}
%\date{}
\author{Sivaramakrishnan Natarajan Ramamoorthy\thanks{Supported by the National Science Foundation under agreement CCF-1420268.} \\ 
	{\small University of Washington, Seattle} \\
	{\small \tt sivanr@cs.washington.edu}
	\and 
	Cyrus Rashtchian \\ {\small University of California, San Diego} \\ {\small \tt crashtchian@eng.ucsd.edu}  }
%\author{}

\maketitle

\begin{abstract}
Recently, Dvir, Golovnev, and Weinstein have shown that sufficiently strong lower bounds for linear data structures would imply new bounds for rigid matrices. However, their result utilizes an algorithm that requires an $NP$ oracle, and hence, the rigid matrices are not explicit.
In this work, we derive an equivalence between rigidity and the systematic linear model of data structures.
For the $n$-dimensional inner product problem with $m$ queries, we prove that lower bounds on the query time imply rigidity lower bounds for the query set itself. 
In particular, an explicit lower bound of $\omega\left(\frac{n}{r}\log m\right)$ for $r$ redundant storage bits would yield better rigidity parameters than the best bounds due to Alon, Panigrahy, and Yekhanin.
We also prove a converse result, showing that rigid matrices directly correspond to hard query sets for the systematic linear model.
As an application, we prove that the set of vectors obtained from rank one binary matrices is rigid with parameters matching the known results for explicit sets. This implies that the vector-matrix-vector problem requires query time $\Omega(n^{3/2}/r)$ for redundancy $r \geq \sqrt{n}$ in the systematic linear model,
improving a result of Chakraborty, Kamma, and Larsen. Finally, we prove a cell probe lower bound for the vector-matrix-vector problem in the  high error regime, improving a result of Chattopadhyay, Kouck\'{y}, Loff, and Mukhopadhyay. 
\end{abstract}

\newpage

\section{Introduction}
A matrix is {\em rigid} if it is far in Hamming distance from low rank matrices; it is {\em explicit} if its entries are computable in polynomial time. A classic result of Valiant proves that explicit rigid matrices imply super-linear lower bounds for linear circuits~\cite{Valiant77}, a major open problem in computational complexity~\cite{shpilka2010arithmetic, viola2009power}. Implications of new lower bounds for communication complexity and other models are also known~\cite{Lokam09, wunderlich2012theorem}.  Unfortunately, the current bounds for explicit matrices are very far from the required parameters~\cite{Friedman93, ShokrollahiSS97}, and natural candidates (e.g., Fourier and Hadamard matrices) have been discovered to be less rigid than desired~\cite{alman2017probabilistic, dvir2017matrix, dvir2019fourier}.
This motivates alternative avenues for constructing rigid matrices.  Recently, multiple connections between data structures and circuits have arisen~\cite{brody2015adapt, Corrigan-GibbsK18, DvirGW19, Viola18}. The premise of these results is that hard problems for these models may shed new light on rigid matrices and circuits. 
We take a similar angle, studying a generic linear problem for a model that resembles a depth-two circuit with linear gates. 

Valiant's result concerns arithmetic circuits computing the linear map $v \mapsto Mv$ for a matrix~$M$. In other words, the circuit computes the inner products between $v$ and the rows of $M$. We study a related data structure problem, the {\em inner product problem}. The task is to preprocess an $n$-bit vector~$v$ to compute inner products $\inner{q}{v}$ over $\F_2$ for queries $q \in Q$, where $Q \subseteq \F_2^n$ is the {\em query set}. This problem generalizes the prefix-sum problem~\cite{gal2007cell} and vector-matrix-vector problem~\cite{ChakrabortyKL18, LarsenW17}.

We consider solving this problem using a restricted data structure model, the {\em systematic linear model}. %During preprocessing, 
This model may only store $v$ verbatim along with a small number $r \ll n$ of {\em redundant} bits, which are the evaluations of $r$ linear functions of $v$.
%linear functions of $v$, where $r$ is the redundant space.
To compute $\inner{q}{v}$ for $q \in Q$, the query algorithm must output a linear function of these $r$ bits along with any $t$ bits of $v$, where $t$ is the {\em query time}.
We motivate this model with a simple upper bound. 
Suppose that the query set $Q$ happens to be close to an $r$-dimensional subspace $U$. More precisely, assume that $\dH{q}{U}  \leq t$ for any $q \in Q$, where $\dH{q}{U} := \min_{u \in U}{ \dH{q}{u}}$ and $\dH{q}{u}$ denotes the Hamming distance. The systematic linear model will store $r$ bits that correspond to inner products between $v$ and some $r$ vectors that form a basis for $U$.
The query algorithm computes $\inner{q}{v}$ by invoking the identity  $\inner{q}{v} = \inner{u}{v} + \inner{q-u}{v}$, using any vector $u \in U$ with $\dH{q}{u} \leq t$. Indeed, the $r$ precomputed bits suffice to determine $\inner{u}{v}$, and at most $t$ bits of $v$ are needed to calculate $\inner{q-u}{v}$.

We observe that rigidity exactly captures the complexity of the inner product problem in the above model.
This connection uses a notion of rigid sets, defined by Alon, Panigrahy and Yekhanin~\cite{AlonPY09}. Our result shows that an efficient algorithm exists in the above model if and only if the query set is not rigid in their sense. Conversely, it is possible to derive new rigidity lower bounds by proving lower bounds for {the} systematic linear model. A parameter of interest is the size of the rigid set, which corresponds to the number of queries in the inner product problem.

Dvir, Golovnev, and Weinstein also demonstrate a connection between rigidity and a different linear model, which is a restriction of the cell probe model~\cite{DvirGW19}. This model stores $s \geq n$ linear functions, and the query algorithm outputs a linear function of $t$ of these $s$ bits.
For the inner product problem with query set $Q$, they show that a lower bound for linear data structures leads to a semi-explicit rigid set. When $|Q|=m$, their result uses a $\mathsf{poly}(m)$ time algorithm that requires access to an $NP$ oracle. 
%This leads to a set $Q'$ consisting of $n'$-dimensional vectors for $n' \leq n$, and a lower bound of time $t$ for space $s=O(n)$ implies that $Q'$ is $(\frac{n'}{2},t/\log n)$-rigid.
Compared to their work, our connection preserves explicitness and offers a two-way equivalence via the systematic linear model. In particular, when $r = \Theta(n)$, a lower bound of $t = \omega(\log m)$ in the systematic linear model implies that $Q$ is rigid with better parameters than known results. Their work requires a lower bound of $t = \omega(\log m \log n)$ against the linear model, and the resulting set is not explicit. 
Our results also extend to show that linear data structure lower bounds lead to explicit rigid matrices. However, compared to the work of Dvir, Golovnev, and Weinstein, we require stronger lower bounds to achieve new rigidity parameters.

As an application of our framework, we provide new results for 
the vector-matrix-vector problem. The task is to preprocess a 0-1 matrix~$M$ to compute $\trans{u}Mv$ when given vectors $u,v$ as the query. The boolean semiring version of this problem has received much recent attention due to connections to the online matrix-vector multiplication conjecture~\cite{HenzingerKNS15}. Moreover, this problem has motivated the study of data structures for a super polynomial number of queries, even when the output is binary~\cite{ChakrabortyKL18,ChattopadhyayKLM18}. Other prior work has either studied binary output problems with $\mathsf{poly}(n)$ queries (see e.g. \cite{Patrascu11, PatrascuT10}) or achieved better lower bounds by looking at multi-output problems (see e.g. \cite{CliffordGL15, Larsen12a}). In general, the vector-matrix-vector problem is a good testbed for proving better data structure lower bounds, because linear algebraic tools could provide new insights.

The $\F_2$ variant of this problem specializes the inner product problem because $\trans{u}Mv$ equals the inner product 
of $u\trans{v}$ and $M$ (viewed as vectors). 
The query set consists of $\sqrt{n}\times\sqrt{n}$ matrices with rank one; its size $m$ satisfies $\log m = \Theta(\sqrt{n})$.
As another contribution, we lower bound the rigidity of this set, and consequently, we obtain a query time lower bound of $\Omega(\frac{n}{r} \log m) = \Omega(n^{3/2}/r)$ for the systematic linear model with redundancy $r \geq \sqrt{n}$. Any asymptotically better lower bounds for this problem (in the systematic linear model) would directly imply that this query set is rigid with better parameters than the currently known results for explicit matrices~\cite{AlonC15, AlonPY09}. 
%We note the approach of \cite{DvirGW19} requires $\mathsf{poly}(m) = \mathsf{exp}(\sqrt{n})$ time to compute an element of the rigid set. 

As a final result, we prove a new cell probe lower bound for the vector-matrix-vector problem, without restrictions on the data structure. Our result
improves the current best lower bound due to Chattopadhyay, Kouck\'{y}, Loff, and Mukhopadhyay~\cite{ChattopadhyayKLM18}. Our lower bound matches the limit of present techniques and achieves the current best time-space trade-off in terms of query set size.

\subsection{Rigid sets, systematic linear model, and the inner product partial function}
 
Throughout, let $m=m(n)$ and $t = t(n)$ and $r = r(n)$ denote positive integers, 
with $m \geq n \geq t,r$. Alon, Panigrahy and Yekhanin defined the following notion of a rigid set~\cite{AlonPY09}.

\begin{definition*}[Rigid Set]
	%Let $n, r, t$ be positive integers such that $r,t < n$. 
	A set $Q \subseteq \F_2^{n}$ is $(r,t)$-rigid if for every subspace $U \subseteq \F_2^{n}$ with dimension at most $r$, some vector $q \in Q$ has Hamming distance at least $t$ from all vectors in $U$, that is, $\dH q U \geq t$.
\end{definition*}
\noindent
We  define $(r', t')$-rigid for non-integral $r', t'$ to mean $(\floor{r'}, \ceil{t'})$-rigid. It will be convenient to equate a set $Q$ with a matrix $M_Q$ by arranging vectors in $Q$ as rows in $M_Q$ in any order.
If $Q$ is $(r,t)$-rigid and $|Q| = m$, then the corresponding matrix $M_Q \in \F_2^{m \times n}$ is rigid in the usual sense: for any rank $r$ matrix $A$, some row in $(M_Q - A)$ contains at least $t$ nonzero entries.
% Specifically, the vectors in $Q$ form the rows of $M_Q$ and the rows of any rank $r$ matrix $A$ reside in some $r$-dimensional subspace.
Hence, we may refer to rigid sets and rigid rectangular matrices interchangeably. A matrix in $\F_2^{m \times n}$ (or a set of $n$-dimensional vectors) is {\em explicit} 
if every entry can be computed in $\mathsf{poly}(n)$ time. 

A random $m \times n$ matrix with $m = \mathrm{poly}(n)$ will be $(\epsilon n , \delta n/\log n )$-rigid with high probability for some constants $\epsilon, \delta \in (0,1)$. 
The key challenge here is to construct explicit rigid matrices, because they provide circuit lower bounds for functions that can be described in polynomial time~\cite{Valiant77}. 
Alon, Panigrahy and Yekhanin~\cite{AlonPY09} followed by Alon and Cohen~\cite{AlonC15} exhibit multiple examples of explicit $m \times n$ matrices that are $(r,t)$-rigid with 
\begin{align}
t\geq \min\left \{ \frac{cn}{r} \log \frac{m}{r},\ n \right\} \label{eqn:rigidity-bestlb}
\end{align} where $m\geq n$ and $c$ is a constant. 
Note that when $r=\epsilon n$, the current best bound is $t = \Omega\left(\log \frac{m}{n}\right)$. For $m = poly(n)$, this amounts to $t = \Omega(\log n)$, exponentially far from the ideal bounds (i.e., matching random constructions).
It is an important open problem to improve the dependence on~$m$ in \Cref{eqn:rigidity-bestlb} and to find other candidate sets that may be rigid with better parameters.   
%For more information on rigidity and applications, we refer the reader to Lokam's survey~\cite{Lokam09}. 

%We show that set rigidity exactly captures the complexity of a certain data structure model. 
Our connection between rigidity and data structures arises via the inner product problem. The task is to preprocess a vector $v \in \F_2^n$ to compute inner products. The queries are specified by $Q \subseteq \F_2^n$, which is called the {\em query set}. The data structure must compute the inner product of $v$ and any $q \in Q$, that is,
$\inner{q}{v} := \sum_{i=1}^n q[i] \cdot v[i] \mod 2,$ where $q[i]$ denotes the $\nth{i}$ coordinate of $q$. 

Consider the following model for solving this problem, known as a {systematic linear data structure}. During preprocessing, the data structure stores $v$ along with the evaluations of $r$ linear functions $\inner{a_1}{v}, \ldots, \inner{a_r}{v}$, where these inner products are single bits, and $a_1,\ldots, a_r$ denote vectors in $\F_2^n$.
To compute the answer on query $q$, the data structure accesses these $r$ bits in addition to any $t$ entries of $v$. That is, the $r$ linear functions are fixed, and the $t$ bits from $v$ may depend on $q$ and the linear functions. Finally, the query algorithm must output a linear function of these $r$ bits and the $t$ entries of $v$. In this fashion it must be able to correctly compute $\inner{q}{v}$ for all queries $q \in Q$. We note that a result of Jukna and Schnitger~\cite{jukna2011min} shows that the $\{a_1,\ldots,a_r\}$ vectors do not depend on $v$ without loss of generality.
Letting $\T(Q,r)$ denote the minimum value~$t$ of the best data structure for this problem (over worst-case $v$), 
we formalize the model as follows.

\begin{definition*}[Systematic Linear Model]
Let $Q \subseteq \F_2^n$ be a set. Define $\T(Q,r)$ to be the maximum over all $v\in \F_2^n$ of the minimum $t$ sufficient to compute the inner product $\inner{q}{v}$ for every $q \in Q$ when only allowed to output a linear function of $r$ precomputed linear functions of $v$ along with any $t$ bits of $v$.
\end{definition*}

Note that the model does not charge the query time for accessing the $r$ precomputed bits, even if $t \ll r$. %This affects the bounds on $t$ when $t \ll r$. 
This coincides with the systematic model studied by Chakraborty, Kamma and Larsen~\cite{ChakrabortyKL18}.
%, which is motivated by upper bounds on the matrix-vector problem~\cite{LarsenW17}. 

\subsection{Equivalence between rigidity and data structures}

We prove that the rigidity of a set $Q$ corresponds to the time complexity $\T(Q,r)$ in the systematic linear data structure model. 
Some aspects of this result are implicit in prior work
~\cite{jukna2011min, pudlak1997boolean}, but no previous work seems to show this exact correspondence.

\begin{theorem}\label{thm:equivalence}
	A set $Q \subseteq \F_2^n$ is $(r,t)$-rigid if and only if $\T(Q,r) \geq t$.
\end{theorem}
\begin{proof}
	We first prove that $\T(Q,r) \geq t$ implies that $Q$ is $(r,t)$-rigid. Assume for contradiction that there is an $r$-dimensional subspace $U$ such that $\dH{q}{U} <t$ for all $q\in Q$. 
	Let $v\in \F_2^n$ be the input data. Store $v$ along with the $r$ bits $\inner{b_1}{v},\ldots,\inner{b_r}{v}$, where $b_1,\ldots,b_r$ form a basis for $U$. For every $q\in Q$, there exists $u_q\in U$ such that $q-u_q$ has Hamming weight less than~$t$. Using the $r$ redundant bits, the algorithm on query $q$ can compute $\inner{u_q}{v}$ by writing $u_q$ in terms of the stored basis vectors. Then, it computes $\inner{q-u_q}{v}$ by accessing fewer than $t$ coordinates of $v$. Since $\inner{q}{v} = \inner{u_q}{v} + \inner{q-u_q}{v}$, we have that $T(Q,r) <t $, which is a contradiction.
	
	We now prove that if $Q$ is $(r,t)$-rigid, then $\T(Q,r) \geq t$. Let $e_1,\ldots,e_n$ denote the standard basis, and let $k = T(Q,r)$ be the query time.
	We show that $k \geq t$. Consider a systematic linear data structure whose redundant bits are given by $\inner{a_1}{v},\ldots,\inner{a_r}{v}$. 
	Let $U$ denote the span of $\{a_1,\ldots,a_r\}$. As $Q$ is $(r,t)$-rigid, there exists $q^*\in Q$ with $\dH{q^*}{U} \geq t$. 
	When $q^*$ is the query, assume that the query algorithm accesses the bits $v_{i_1},\ldots,v_{i_{k}}$ for indices $i_1,\ldots, i_k$ to compute $\inner{q^*}{v}$. Now, define $U'$ to be the span of $\{a_1,\ldots,a_r,e_{i_1},\ldots,e_{i_{k}}\}$. Observe that all points in $U'$ are at distance at most $k$ from $U$.  Thus, $\dH{q^*}{U} \leq  \dH{q^*}{U'} + k$.
We will show that $\dH{q^*}{U'}=0$, which implies that $k\geq t$.
We claim that if $\dH{q^*}{U'} \geq 1$, then the query algorithm makes an error. Since $\dH{q^*}{U'} \geq 1$, there exists a vector $y$ with $\inner{y}{q^*} = 1$. Moreover, this vector can be taken to be orthogonal to $U'$ so that $\inner{y}{x} = 0$ for every $x\in U'$. In other words, for every $x\in U'$ we have $\inner{y+v}{x} = \inner{y}{x} + \inner{v}{x} = \inner{v}{x}$. Hence, the query algorithm sees the same values on input data $y+v$ and $v$ because it only accesses the input via vectors in $U'$, and we have $x \in U'$. 
Thus, the algorithm on query $q^*$ must err either on input $y+v$ or $v$ because $\inner{q^*}{y+v} \neq \inner{q^*}{v}$.
\end{proof}

\subsection{Relationship to the cell probe model and other models}
\label{sec:models}

The systematic linear model specializes the {\em systematic model}~\cite{ChakrabortyKL18, gal2007cell}. The latter model still stores the input data $x \in \F_2^n$ verbatim, and it also stores $r < n$ bits that can be precomputed from $x$, where these need not be linear functions of the input data. The query time is $t$ if the query algorithm reads at most $t$ bits from $x$ to compute a query. The output can also be an arbitrary function of these $t$ bits along with the $r$ precomputed bits. The systematic linear model only makes sense for linear queries, whereas the systematic model applies to arbitrary query functions.

Yao's {\em cell probe model} is the most general data structure model~\cite{Yao81}. On input data $x \in \F_2^n$, the data structure stores $s$ cells, containing $w$ bits that are arbitrary functions of $x$. Here, $w$ is the \emph{word size} and $s$ is the \emph{space}. The {\em query time} is $t$ if the algorithm accesses at most $t$ cells to answer any query about $x$ from a set of $m$ possible query functions. There is a rich collection of lower bounds for this model (see e.g.~\cite{Ajtai88,FredmanS89,
	Larsen12a, MiltersenNSW95, PanigrahyTW10, Patrascu11, PatrascuD06}). The best lower bounds known are of the form 
\begin{align}
t \geq \min\left \{
\frac{c\log \frac{m}{n}}{\log \frac{sw}{n}}, \ \frac{cn}{w}\right\},
\label{eqn:ds-bestlb}
\end{align} 
where $m\geq n$ is the number of queries and $c$ is a constant. It is a long-standing problem to prove that $t = \omega(\log m)$ for any explicit problem, even in the linear space regime $s \cdot w = O(n)$.

A special case of the cell probe model is the {\em linear model}~\cite{Agarwal97, DvirGW19}. The latter model stores $s \geq n$ linear functions of $x$ (implicitly $w =1$ is fixed). The query time is $t$ if the query algorithm reads at most~$t$ of these~$s$ bits to compute a query. The output is restricted to be a linear function of these $t$ bits. A distinguishing aspect between linear and systematic linear is that in the latter model, the query algorithm is not charged for accessing the $r$ precomputed bits. In \Cref{sec:ds-rigidity}, we compare the linear and systematic linear models in the context of rigidity and previous work~\cite{DvirGW19}.

\paragraph{Equivalences all the way down.}
We note that the systematic data structure model is identical to the common bits model defined by Valiant~\cite{Valiant92}. 
Corrigan-Gibbs and Kogan~\cite{Corrigan-GibbsK18} 
demonstrate a relationship between the common bits model and a variant of the systematic model defined by Gal and Miltersen~\cite{gal2007cell}. The common bits model is nothing but a certain depth two circuit, and the systematic linear model is simply the common bits model with the restriction that the common bits and output gates are linear functions~\cite{pudlak1997boolean}. Hence, in language of data structures, the linearization conjecture of Jukna and Schnitger posits that the systematic linear model is asymptotically as powerful as the systematic model for answering linear queries~\cite{jukna2011min}.

\subsection{The vector-matrix-vector problem}

We now define the vector-matrix-vector problem, which we call ``the $\umv$ problem'' for short. Let $n$ be a perfect square.
%In the standard formulation, the input data is a $\sqrt{n} \times \sqrt{n}$ boolean matrix $M$. 
After preprocessing a matrix~$M \in \F_2^{\sqrt{n} \times \sqrt{n}}$, the goal is to output the binary value $\trans{u}Mv$ for vectors $u,v \in \F_2^{\sqrt{n}}$. 
It will be convenient to consider a $\sqrt{n}\times \sqrt{n}$ matrix as an $n$-bit vector  $\ve(M)$ by concatenating consecutive rows. 
More formally, let $x = \ve(M)$, and for $i\in \{1,2,\ldots,n\}$, set $x[i] = M[a,b]$, where $a$ and $b$ satisfy $i= (a-1)\sqrt{n} + b$ and $a, b \in \{1,2,\ldots,\sqrt{n}\}$. Then, $\trans{u}Mv = \inner{\ve(u\trans{v})}{\ve(M)}$.
In this way we consider the $\umv$ problem a special case of the inner product problem. The query set is the collection of rank one binary matrices. Let $\Upsilon \subseteq \F_2^n$ denote the set of vectors obtained from rank one binary matrices via $M \mapsto \ve(M)$, that is,
\begin{align}
\Upsilon :=  \left\{\ve(u\trans{v}) \relmiddle| u,v \in \F_2^{\sqrt{n} \times \sqrt{n}}\right\} \subseteq \F_2^n. \label{def:sigma-rankone}
\end{align}
This set has size $|\Upsilon| = 2^{2\sqrt{n}} - 2^{\sqrt{n}+1} + 1$.
%, and this is the number of queries for the $\umv$ problem.

A classic result of Artazarov, Dinic, Kronrod and Faradzev~\cite{ArtazarovDKF70} provides a data structure with space $s=\mathsf{poly}(n)$, word size $w=O(\log n)$, and time $t= O(n/\log n)$. In fact, this algorithm operates in the linear cell probe model. 
It is a central open question to determine whether $t = \Omega(n)$ is necessary in linear space regime, that is, when $s \cdot w = O(n)$. 

The current best cell probe lower bound for the $\umv$ problem is due to  Chattopadhyay, Kouck\'{y}, Loff, and Mukhopadhyay \cite{ChattopadhyayKLM18}. Moreover, their lower bound holds for a randomized model with high error. For constants $c$ and $c'$, they prove that if for every matrix $M$ and every query $u\trans{v}$, the query algorithm correctly computes $\trans{u}Mv$ with probability at least $\frac{1}{2} + \frac{1}{2^{c'\sqrt{n}}}$, then 
\begin{equation}\label{eqn:CKLM-lb}
t \geq \min\left\{\frac{c\sqrt{n}}{\log \frac{sw}{\sqrt{n}}}, \ \frac{cn}{w} \right\}
\end{equation} 

Better lower bounds for the $\umv$ problem are known in the systematic model. Chakraborty, Kamma, and Larsen~\cite{ChakrabortyKL18} prove that $t$ and $r$ must satisfy $t\cdot r = \Omega(n^{3/2}/\log n)$ as long as $r \geq \sqrt{n}$. In the case of $r \leq \sqrt{n}$, they prove that $t = \Omega(n / \log n)$. As the systematic model subsumes the linear version of this model, combining their result with \Cref{thm:equivalence} implies that $\Upsilon$ is $(r,t)$-rigid with 
\begin{align}\label{eqn:CKL-rigidity-lb}
t = \Omega\left(\frac{n^{3/2}}{\max\{\sqrt{n},r\} \cdot \log n}\right).
\end{align}

\subsection{New results on the rigidity of $\Upsilon$ and the cell probe complexity of the $\umv$ problem}

%\paragraph{The rigidity of rank one matrices.}

We lower bound the rigidity of $\Upsilon$, defined in \Cref{def:sigma-rankone}. This also implies a lower bound in the systematic linear model. 
The proof is inspired by a result of Alon, Panigrahy, and Yekahnin~\cite{AlonPY09}. 
 
\begin{theorem}\label{thm:rank1-rigidity} 
	%Let $\Upsilon$ be the matrix defined in (\ref{def:sigma-rankone}). 
	Let $n \geq 1024$. The set $\Upsilon \subseteq \F_2^{n}$ of rank one  matrices is $(r,t)$-rigid with $t\geq \frac{n^{3/2}}{128\cdot \max\{\sqrt{n},r\}}$.
\end{theorem}
We improve the prior bound in \Cref{eqn:CKL-rigidity-lb} by an $\Omega(\log n)$ factor.
For example, when $r\leq\sqrt{n}$, then $t = \Omega(n)$, and   
when $r = n/2$, then $t = \Omega\left(n^{1/2}\right)$.
\Cref{thm:rank1-rigidity} matches \Cref{eqn:rigidity-bestlb}, the current best bound for explicit rigid sets. We do not know whether there is a subspace $U$ of linear dimension such that all elements of $\Upsilon$ are at distance $o(n)$ from $U$ (unlike for some set rigidity results, where the bounds are tight).
As a corollary of \Cref{thm:equivalence}, we immediately get that $$\T(\Upsilon,r) \geq \frac{n^{3/2}}{128\cdot \max\{\sqrt{n},r\}}.$$ In other words, we prove a lower bound for the $\umv$ problem in the systematic linear model that improves the prior bound by an $\Omega(\log n)$ factor. %Compared to the previous work, our proof is simple and linear algebraic.  
The proof of \Cref{thm:rank1-rigidity} appears in \Cref{sec:rank1-rigidity}.

%\paragraph{New general data structure lower bounds.}
We also prove a general cell probe lower bound for the uMv problem in the high error regime. Our result improves the previous lower bound in \Cref{eqn:CKLM-lb}. For example, in the linear space regime, when $s\cdot w=O(n)$, we show that $t = \Omega(\sqrt{n})$ while the prior result gives only $t = \Omega(\sqrt{n}/\log n)$.

\begin{theorem}
	\label{thm:uMv-general}
	%Let $n > 32$ be an integer. 
	Let $M \in \F_2^{\sqrt{n} \times \sqrt{n}}$ be a matrix. If a randomized data structure with space $s$, word size $w$, and time $t$ correctly computes queries for the $\umv$ problem
	with probability at least $\frac{1}{2} + \frac{1}{2^{\sqrt{n}/64}}$, then% We then have that 
	\[
	t \geq \min\left\{\frac{c\sqrt{n}}{\log \frac{s\alpha}{n}}, \ \frac{cn}{\alpha} \right\}\]
	where $0 < c \leq 1/36$ is a universal constant and $\alpha := 2(w + \log \frac{sw}{n})$.
\end{theorem}
{
The prior work utilizes a general lifting result for two-way communication complexity from parity decision trees~\cite{ChattopadhyayKLM18}. To obtain the improved bound,  we use a variant of the {cell sampling} technique~\cite{Larsen12,PanigrahyTW10} combined with a reduction to a new lower bound on one-way communication (via discrepancy). The modifications over standard techniques are needed to handle the high error regime for a binary output problem. 
We note that a recent result of Larsen, Weinstein and Yu also uses one-way communication to prove lower bounds for binary output problems for dynamic data structures~\cite{LarsenWY18}. However, their method seems limited to only handling zero error query algorithms. 
The proof of \Cref{thm:uMv-general} appears in \Cref{sec:lb-proofs}. Specifically, see \Cref{lem:cellsampling} in \Cref{sec:lb-proofs} for the variant of cell sampling and see \Cref{thm:uMv-oneway} in \Cref{sec:uMv-oneway} for the discrepancy argument.}

\section{Linear Data Structures and Rigidity}
\label{sec:ds-rigidity}

In this section, we relate linear data structures and rigidity. As linear data structures are a special case of the cell probe model, we may obtain rigidity lower bounds from strong enough static data structure lower bounds (when the queries are linear). We also compare with Dvir, Golovnev, and Weinstein, who exhibit a similar connection~\cite{DvirGW19}. 
%The main difference is that our results will be for explicit matrices, whereas theirs require a $P^{NP}$ algorithm. 
We first provide some notation.  
\begin{definition}
Let $Q \subseteq \F_2^n$ be a set. Define $\mathsf{LT}(Q,s)$ to be the maximum over all $v\in \F_2^n$ of the minimum $t$ sufficient to compute the inner product $\inner{q}{v}$ for every $q \in Q$ when the query algorithm's output is a linear function of $t$ bits chosen from the $s$ precomputed linear functions of $v$.
\end{definition}

Table~\ref{table:result-compare} provides a glimpse of our results on linear data structures along with a comparison to \cite{DvirGW19}. 
Recall that a set $Q\subseteq \F_2^n$ is explicit if each coordinate of an arbitrary element of the set can be computed in $\mathsf{poly}(n)$ time.
The prior work shows that sufficiently strong lower bounds against linear data structures will imply \emph{semi-explicit} rigid sets.  A bit more formally, consider a data structure query set $Q\subseteq \F_2^n$ of size $m$ for the inner product problem. They show the following: If $\mathsf{LT(Q,c\cdot n)} \geq t$ for some constant $c$, then there is a $(n'/2,t/\log n)$-rigid set $Q'$ of size at most $m$ contained in $\F_2^{n'}$, where $n' \geq t$. However, the set $Q'$ is only semi-explicit in that it is in $P^{NP}$ -- every element can be computed by a $\mathsf{poly}(m)$ time algorithm with access to an $NP$ oracle. 

We now summarize a few differences between our work and \cite{DvirGW19}. Our result proves that polynomial lower bounds on the query time imply the existence of an explicit rigid set, which is in contrast to semi-explicit sets obtained by \cite{DvirGW19}. On the other hand, explicitness comes with a cost; when $m = \mathsf{poly}(n)$, we need much stronger data structure lower bounds to produce explicit rigid sets. When $m \gg \mathsf{poly(n)}$, the algorithm of \cite{DvirGW19} takes $\mathsf{poly}(m)$ time with access to an NP oracle to compute an element of the semi-explicit rigid set. For problems such as the $\umv$ problem, this is super polynomial time.
The rest of this section concerns proving the following theorem, which implies all of our results in Table~\ref{table:result-compare}. 
\begin{theorem}\label{thm:linearlb-rigidity}
Let $k = \mathsf{LT}(Q,3n/2)$ and let $Q \subseteq \F_2^{n}$ of size $m$ be an explicit query set. There exists a set $Q'\subseteq \F_2^{k}$ with size at most $m \cdot \ceil{\frac{n}{k}}$, whose elements can be computed in $\mathsf{poly}(n)$ time. Moreover, if $k \geq 2\sqrt{n}$, then $Q'$ is explicit and $\left(\frac{k}{2},\frac{k^2}{4n}\right)$-rigid.
\end{theorem}

Note that for every $s\geq 3n/2$, we have that $\mathsf{LT}(Q,3n/2) \geq \mathsf{LT}(Q,s)$. Hence, a sufficiently strong lower bound on $\mathsf{LT}(Q,s)$ for any $s\geq 3n/2$ will imply a rigidity lower bound.
The following corollary shows the consequence of Theorem~\ref{thm:linearlb-rigidity} for specific values of $k$. 

\begin{corollary}\label{cor:linearlb-rigidity}
Let $k = \mathsf{LT}(Q,3n/2)$ and let $Q \subseteq \F_2^{n}$ of size $m$ be an explicit query set. There exists a set $Q'\subseteq \F_2^{k}$ with size at most $m \cdot \ceil{\frac{n}{k}}$, whose elements can be computed in $\mathsf{poly}(n)$ time. Moreover,
\begin{enumerate}
\item[(a)] If $k = \omega\left(\sqrt{n\log m}\right)$, then $Q'$ is explicit and $\left(k/2,\omega(\log m)\right)$-rigid.
\item[(b)] If $k = \Omega\left(n^{(1+\delta)/2}\right)$ for some $\delta >0$, then $Q'$ is explicit and $\left(k/2,\Omega\left(n^{\delta}\right)\right)$-rigid.
\end{enumerate}
\end{corollary}

\Cref{cor:linearlb-rigidity}(a) explains the first and last rows in \Cref{table:result-compare}, and \Cref{cor:linearlb-rigidity}(b) explains the middle row. 
Using \Cref{cor:linearlb-rigidity}(a) applied to $\Upsilon$ with $m = 2^{2\sqrt{n}}-2^{\sqrt{n}+1}+1$, we obtain that a lower bound of $\mathsf{LT}(\Upsilon,3n/2) \geq \omega(n^{3/4})$ would imply the existence of an explicit set $Q' \subseteq \F_2^k$ of size $2^{O(\sqrt{n})}$ that is $(k/2,\omega(\sqrt{n}))$-rigid.
We note that it is an open question to prove $\mathsf{LT}(\Upsilon,3n/2) \geq \omega(\sqrt{n})$.

\begin{table}[t]
	\ra{1.5}
	\adjustbox{max width = \textwidth}{
		\centering
		\begin{tabular}{@{}lllll@{}}\toprule
			$m$ vs $n$&{$k=\mathsf{LT}(Q,3n/2)$}& Rigidity Bounds&Explicitness&Reference\\
			\midrule
			\multirow{2}{*}{$m = n^c$} &  $k = \omega\left(\sqrt{n \log n}\right)$  &$\left(k/2,\omega\left(\log n\right)\right)$-rigid& $\mathsf{poly}(n)$ time & This work\\ 
			& $k = \omega\left(\log^2 n\right)$   &$\left(k/2,\omega\left(\log n\right)\right)$-rigid&  $\mathsf{poly}(n)$ time + $NP$ oracle calls & \cite{DvirGW19} \\ 
			%& & &+$NP$ oracle calls&\\
			\cmidrule{1-5}
			\multirow{2}{*}{$m = n^c$} &  $k = \Omega\left(n^{(1+\delta)/2}\right)$  &$\left(k/2,\Omega\left(n^{\delta}\right)\right)$-rigid& $\mathsf{poly}(n)$ time & This work\\ 
			& $k = \Omega\left(n^{\delta}\log n\right)$  &$\left(k/2,\Omega\left(n^{\delta}\right)\right)$-rigid&  $\mathsf{poly}(n)$ time + $NP$ oracle calls & \cite{DvirGW19}\\ 
			%& & &+$NP$ oracle calls&\\
			\cmidrule{1-5}
			\multirow{2}{*}{$m = 2^{c\sqrt{n}}$} &   $k = \omega\left(n^{3/4}\right)$ &$\left(k/2,\omega\left(\sqrt{n}\right)\right)$-rigid& $\mathsf{poly}(n)$ time & This work\\ 
			& $k = \omega\left(\sqrt{n}\cdot \log n\right)$  &$\left(k/2,\omega\left(\sqrt{n}\right)\right)$-rigid&  $\mathsf{poly}\left(2^{\sqrt{n}}\right)$ time + $NP$ oracle calls & \cite{DvirGW19} \\ 
			%& & &+$NP$ oracle calls&\\
			\bottomrule
	\end{tabular}}
	\caption{Comparison with \cite[Theorem 7.1]{DvirGW19}: Let $Q \subseteq \F_2^{n}$ of size $m$ be a query set, $c \geq 1$ and $\delta > 0$ be constants, and let {$k=\mathsf{LT}(Q,3n/2)$}. The second column states the lower bound on {$\mathsf{LT}(Q,3n/2)$} that implies existence of rigid sets whose parameters are given in the third column. All rigid sets have size at most $\mathsf{poly}(m)$ and are contained in $\F_2^{k}$.}
	\label{table:result-compare}
\end{table}

\subsection{Proof of \Cref{thm:linearlb-rigidity}}
We already know the equivalence between systematic linear data structures and rigidity (from \Cref{thm:equivalence}). Therefore, it is sufficient to design a linear data structure from a systematic linear data structure to relate the former with rigidity. 
\begin{proposition}\label{prop:linear-succinct}
	Let $Q \subseteq \F_2^{n}$ be a query set. If $\ \T(Q,r) \leq t$, then $\mathsf{LT}(Q,n+r) \leq t+r$.
\end{proposition}
\begin{proof}
Let $v \in \F_2^{n}$ be the input data, and let $\inner{a_1}{v},\ldots,\inner{a_r}{v}$ be the $r$ redundant bits stored by the systematic linear data structure. We now describe a linear data structure for $Q$ with space $n+r$ and query time $t+r$. The data structure stores $\inner{a_1}{v},\ldots,\inner{a_r}{v},\inner{e_1}{v},\ldots,\inner{e_n}{v}$, where $e_1,\ldots,e_n$ are the standard basis vectors. The query algorithm on $q\in Q$ first accesses $\inner{a_1}{v},\ldots,\inner{a_r}{v}$ and then simulates the query algorithm of the systematic linear data structure on $q$. Since the systematic linear data structure accesses at most $t$ bits from $\inner{e_1}{v},\ldots,\inner{e_n}{v}$, we can conclude that the query time is at most $t+r$.
\end{proof}

We prove that if a set contained in a $n$-dimensional space is $(r,t)$-rigid, then there is another $(r, tr/n)$-rigid set which is contained in a $2r$-dimensional space. 
\begin{lemma}\label{lem:submatrix-rigidity}
Let $r,n$ be positive integers. If $S \subseteq \F_2^{n}$ is $(r,t)$-rigid of size $m$, then there is a set $S' \subseteq \F_2^{2r}$ of size at most $m \cdot \ceil{\frac{n}{2r}}$ that is $(r,tr/n)$-rigid. Moreover, if $S$ is explicit, then each element of $S'$ can be computed in $\mathsf{poly}(n)$ time.
\end{lemma}
\begin{proof}
Let $k = \floor{\frac{n}{2r}}$ and define $S_1,\ldots,S_{k} \subseteq \F_2^{2r}$ by
	\[S_i = \left\{\left(s[2r \cdot (i-1)+1],\ldots,s[2r\cdot i]\right) \mid s\in S\right\}\] for each $i\in \{1,2,\ldots,k\}$. 
	Additionally, if $n/2r$ is not an integer, then define 
	\[S_{k+1} = \left\{\left(s[2r \cdot k+1], \ldots, s[n], 0,\ldots,0\right) \mid s\in S\right\} \subseteq \F_2^{2r};\] 
	otherwise set $S_{k+1} = \emptyset$. Define $S' = \bigcup_{i=1}^{k+1} S_i$.
	We claim that $S'$ is $(r,tr/n)$-rigid. Indeed, for the sake of contradiction assume that there is a subspace $V$ in $\F_2^{2r}$ of dimension $r$ such that all points in $S'$ are at a distance less than $tr/n$ from $V$. 
	Consider the subspace $\{(v,v,\dots,v) \mid v\in V\} \subseteq \F_2^{2r \cdot (k+1)}$ and project it to the first $n$ coordinates. Call this subspace $V'$, which has dimension $r$. Now, the distance of each point in $S$ from $V'$ is less than $\frac{tr}{n} \cdot \ceil{\frac{n}{2r}} < t$, which is a contradiction. 
	
	Regarding the explicitness of $S'$, it is clear that all coordinates of an element of $S'$ correspond to some coordinate of a specific element of $S$. Since $S$ is explicit, we can infer that each element of $S'$ can be computed in $\mathsf{poly}(n)$.
\end{proof}

\begin{proof}[Proof of \Cref{thm:linearlb-rigidity}]
	Since $\mathsf{LT}(Q,3n/2) = k$ and $k\leq n$, \Cref{prop:linear-succinct} implies that $\T(Q,k/2) \geq k/2$. Therefore by \Cref{thm:equivalence}, we can conclude that $Q$ is $(k/2,k/2)$-rigid. \Cref{lem:submatrix-rigidity} implies that there exists a set $Q'$ that is $\left(\frac{k}{2},\frac{k^2}{4n}\right)$-rigid and the size of $Q'$ is at most $m\cdot \ceil{\frac{n}{k}}$.
	Moreover, every element of $Q'$ can be computed in time $\mathsf{poly}(n)$. Since $k/2 \geq \sqrt{n}$, we can conclude that $Q'$ is explicit.
\end{proof}

\section{Rigidity Lower Bounds for the Set of Rank One Matrices}
\label{sec:rank1-rigidity}
Before proving \Cref{thm:rank1-rigidity}, we present preliminaries.
Recall two standard binomial estimates:
\begin{proposition}
	\label{prop:binom}
	For integers $0\leq k \leq \ell$, 
	\begin{enumerate}
		\item  $\log \binom{\ell}{k} \leq k \cdot \log \frac{e\ell}{k}$.
		\item  if $k \leq \ell/16$, then $\sum_{i=0}^k \binom{\ell}{i} \leq 2^{\ell/4}$.
	\end{enumerate}
\end{proposition}

%\subsection{Preliminaries on Distance and Rigidity}
We will need a useful property about the distance of a point from a subspace. 
\begin{lemma}
\label{lem:subspace-additive}
Let $V \subseteq \F_2^{\ell}$ be a subspace. For $u_1,u_2 \in \F_2^{\ell}$, 
$\dH{u_1+u_2}{V} \leq \dH{u_1}{V}+\dH{u_2}{V}.$
\end{lemma}
\begin{proof}
Let $u_1',u_2' \in V$ be the points in $V$ closest to $u_1$ and $u_2$ respectively. Since $u_1'+u_2' \in V$, we have
$$
d_{H}(u_1+u_2,V) \leq d_{H}(u_1+u_2,u_1'+u_2') = d_{H}(u_1+u_2,u_1'+u_2').
$$
Note that $d_{H}(u_1+u_2,u_1'+u_2')$ is the number of ones in $u_1+u_2+u_1'+u_2'$, which is at most the sum of the number of ones in $u_1+u'_1$  and $u_2+u'_2$. Therefore, 
\[d_{H}(u_1+u_2,u_1'+u_2') \leq d_{H}(u_1,u_1') + d_{H}(u_2,u_2')= d_{H}(u_1,V)+d_{H}(u_2,V). \hfill \qedhere
\]
\end{proof}

A simple counting argument establishes the existence of a point that is far away in Hamming distance from a collection of large sized sets.
\begin{lemma}\label{lem:subspace-distance}
Let $V_1,\ldots,V_k$ be subsets of $\F_2^\ell$, each of size at most $2^{\ell/2}$. If $k < 2^{\ell/4}$, then there is a vector $v\in \F_2^\ell$ such that the Hamming distance of $v$ from each $V_i$ is at least $\ell/16$.
\end{lemma}
\begin{proof}
For every $i\in [k]$, define $\mathsf{B}(V_i,\ell/16) = \left| \left\{v \in \F_2^\ell \relmiddle|  \dH{v}{V_i} < \ell/16 \right\}\right|.$
For any  $u\in V_i$, the number of vectors in $\F_2^\ell$ at a distance less than $\ell/16$ from $u$ is at most $\sum_{j=0}^{\ell/16} \binom{\ell}{j} \leq 2^{\ell/4}$, where the inequality follows from \Cref{prop:binom}.
%$\binom{n}{n/16} \cdot 2^{n/16}$. 
Hence $\mathsf{B}(V_i,\ell/16) \leq |V_i| \cdot 2^{\ell/4} = 2^{3\ell/4}.$
Since
$$
\sum_{i=1}^{k} \mathsf{B}(V_i,\ell/16) \leq k \cdot 2^{3\ell/4} < 2^\ell,
$$
there is a $v\in \F_2^\ell$ such that $\dH{v}{V_i} \geq \ell/16$ for every $i \in [k]$.
\end{proof}

\subsection{Proof of \Cref{thm:rank1-rigidity}}
Let $V$ be any $r'$-dimensional subspace of $\F_2^n$, where $r' \geq r$ is the smallest positive integer divisible by $\sqrt{n}$. 
We first define the inverse of $\ve(\cdot)$. For every $v\in \F_2^n$, define $\m(v)$ to be the matrix obtained by splitting $v$ into $\sqrt{n}$ length consecutive blocks and stacking each of these blocks to form a $\sqrt{n} \times \sqrt{n}$ matrix. Formally, $\m(v) \in \F_2^{\sqrt{n} \times \sqrt{n}}$ is such that $\m(v)[a,b] = v[(a-1)\sqrt{n}+j]$ for every $a,b \in [\sqrt{n}]$. Note that $\ve(\m(v)) = v$.

We provide a brief outline of the proof of \Cref{thm:rank1-rigidity}. The first step of the proof is to produce a vector in $v$ that is at a distance of $\Omega(n)$ from $V$ and $\m(v)$ is low rank. The rank being low is helpful as we can express $\m(v)$ as the sum of a small number of rank one matrices.  \Cref{lem:subspace-additive} will then imply the existence of a rank one matrix that is far away from $V$. 
If we only cared about the existence of a vector that is far away from $V$, \Cref{lem:subspace-distance} would suffice.
% that is at a distance of $\Omega(n)$ from $V$. 
{To ensure that simultaneously the rank is small, we first project $V$ on to $n/2r'$ coordinates indexed by consecutive blocks each of length $2r'$.} Then we find a vector $v'\in \F_2^{2r'}$ that is far away from all the projections, which is still guaranteed by \Cref{lem:subspace-distance}. Concatenating $v'$ with itself $2r'$ times has the property that its corresponding matrix is low rank.

Let $k = \max\left\{ \floor{\frac{n}{2r'}},1\right\}$. The goal is to find a $v\in \F_2^n$ such that $\dH{v}{V} \geq k\cdot r'/8$ and the rank of $\m(v)$ is at most $2r'/\sqrt{n}$.
If $\floor{\frac{n}{2r'}} \geq 1$, then define $S_1,\ldots,S_{k}$ such that 
\[S_i = \left\{(i-1)\cdot 2r' + 1,\ldots, i \cdot 2r'\right\}\] 
for $i\in [k]$; otherwise, define $S_1 = [n]$. By definition, the dimension of $V_{S_i}$ is at most $r' = |S_i|/2$, for every $i\in [k]$. Since $r'\geq \sqrt{n}$ and $n \geq 1024$, we can infer that $k \leq 2r'$ and $2r' < 2^{r'/2}$. 
\Cref{lem:subspace-distance} implies the existence of a $v' \in \F_2^{2r'}$ with the property that $\dH{v'}{V_{S_i}} \geq r'/8$ for every $i \in [k]$. Now define $v\in \F_2^n$ by
$$
v[i]=
\begin{cases}
v'\left[i \ \mathrm{mod}\ 2r'\right]&\text{if }i \leq k \cdot 2r' \text{ and } i \ \mathrm{mod}\ 2r'\neq 0 ,\\
v'\left[2r'\right] &\text{if }i \leq k \cdot 2r' \text{ and } i \ \mathrm{mod}\  2r'=0,\\
0 &\text{if } i >2kr',
\end{cases}
$$
for all $i\in [n]$. In words, $v$ is the length $n$ vector that is the concatenation of $k$ copies of $v'$ along with the vector of zeros of length $n-2kr'$.
By the choice of $v$, 
we get that,
$$
\dH{v}{V} \geq \sum_{i=1}^{k} \dH{v}{V_{S_i}} \geq k \cdot r'/8 .
$$
Moreover, the rank of $\m(v)$ is at most $\frac{2r'}{\sqrt{n}}.$
Therefore we can express
\[\m(v) = \sum_{i=1}^{2r'/\sqrt{n}} a_i \trans{b_i}, \]
for some $a_1,b_1,\ldots,a_{\frac{2r'}{\sqrt{n}}}, b_{\frac{2r'}{\sqrt{n}}} \in \F_2^{\sqrt{n}}$. 
By \Cref{lem:subspace-additive}, we know that 
\[\dH{v}{V} \leq \sum_{i=1}^{2r'/\sqrt{n}} \dH{\ve(a_i\trans{b_i})}{V}.\]
Hence there exists an $i\in \left[\frac{2r'}{\sqrt{n}}\right]$ such that $\dH{\ve(a_i\trans{b_i})}{V} \geq \frac{\sqrt{n}\cdot k}{16} \geq \frac{n^{3/2}}{64r'}$. The observation that $r' \leq 2\max\{\sqrt{n},r\}$ completes the proof of the theorem.

\begin{remark}
[Extension to strong rigidity] 
Alon and Cohen~\cite{AlonC15} defined the notion of \emph{strong} rigidity; a set $Q \subseteq \F_2^{n}$ is $(r,t)$-strongly rigid if for every subspace of $\F_2^{n}$ of dimension at most $r$, the average distance of all the points to the subspace is at least $t$. For strong rigidity, the best lower bounds {known} for explicit sets are also of the form given in \Cref{eqn:rigidity-bestlb}. We can show that $\Upsilon$ is $(r,t)$-strongly rigid with $t \geq \Omega\left(\frac{n^{3/2}}{\max\{\sqrt{n},r\}}\right)$, matching the {best strong rigidity bounds known for explicit sets}. We sketch the proof here. We know that 
\[ \trans{u}M v+\trans{(u+e_i)} M v +  \trans{u}M(v+e_j)+
\trans{(u+e_i)}M(v+e_j) = \trans{b_i}Mb_j,\]
where $u,v\in \F_2^{\sqrt{n}}$ and $e_1,\ldots,e_{\sqrt{n}}$ are standard basis vectors in $\F_2^{\sqrt{n}}$. This fact can be used to prove that the matrix $M_\Upsilon$ corresponding to the set $\Upsilon$ is a generator matrix of a 4-query {locally decodable code} that {tolerates} a constant fraction of errors. A result of \cite[Theorem 6]{DvirGW19} shows that \Cref{thm:rank1-rigidity} and the locally decodable code property of $M_\Upsilon$ imply the strong rigidity of $\Upsilon$. 
\end{remark}
\section{Cell Probe Lower Bounds for the $\umv$  Problem}
\label{sec:lb-proofs}

%\subsection{Overview of Our Techniques}
%(The first paragraph can be moved to just after the statement of the theorem).
%
%In addition to obtaining better bounds, our proof is considerably simpler.  \cite{ChattopadhyayKLM18} proved their lower bound using a \emph{lifting} theorem, which is of independent interest. Our proof uses a variant of the  \emph{cell sampling} technique~\cite{PanigrahyTW10,Larsen12} along with a reduction to \emph{one way communication complexity}. Such a reduction is reminiscent of the ideas used by Larsen, Weinstein and Yu~\cite{LarsenWY18} in the context of dynamic data structure lower bounds. 
We know of two techniques for proving cell probe lower bounds matching \Cref{eqn:ds-bestlb}. One is a technique of P{\v a}tra{\c s}cu and Thorup~\cite{PatrascuT10} who combined the communication complexity simulation of Miltersen~\cite{Miltersen94} with multiple queries on the same input data. 
The other is the technique we use, which is based on cell sampling. Cell sampling typically requires one to work with large sized fields in order to handle errors. This large field size is needed to encode a large subset of the correctly computed queries using a small subset of cells.
Here, we avoid encoding the subset of queries by a reduction to one-way communication complexity.

\paragraph{Proof outline for \Cref{thm:uMv-general}.} 
%We give a brief outline of the proof of \Cref{thm:uMv-general}. 
By Yao's min-max principle, it suffices to prove a lower bound on deterministic data structures. The hard distribution on the input data $M$ and query $u\trans{v}$ we use is given by sampling $M,(u,v)$ uniformly and independently at random. 
We prove the theorem by contradiction, and we start by assuming that the query time is small. The proof is carried out in three steps.
First, modify the data structure so that for every $M$, the fraction of queries correctly computed is at least $1/2$. This modification only increases the query time and space by $1$, and it can only increase the overall probability of the query algorithm being correct.
Second, for a given $M$, we use a variant of cell sampling (see \Cref{lem:cellsampling}) to obtain a small subset of cells $S$ and a large subset of queries $Q'$ such that all queries in $Q'$ can be computed by only accessing cells in $S$. 
Moreover, 
\begin{align*}
&\Pr\left[\text{query algorithm correctly computes }\trans{u}Mv \relmiddle| u\trans{v}\in Q'\right] \\
&\approx \Pr\left[\text{query algorithm correctly computes }\trans{u}Mv\right].
\end{align*} 
Third, we show that $S$ can be used to design an efficient protocol for the following communication game: Alice's input is $M$ and Bob's input is $u\trans{v}$, and the goal is for Bob to correctly compute $\trans{u}Mv$ on a sufficiently good fraction of the inputs after receiving a message from Alice. 

We now describe the protocol (see \Cref{prot:ds-simulation}). Alice sends the locations and contents of $S$. This ensures that Bob correctly computes $\trans{u}Mv$ on a large fraction of queries in $Q'$. 
Alice also communicates the majority value of $\trans{u}Mv$ for $u\trans{v} \notin Q'$ so that Bob is correct on half of his possible inputs that are not in $Q'$. 
Overall, Bob's output is correct on a sufficiently good fraction of all $M,(u,v)$. 
Since we have assumed that the query time is small, we are able to show that Alice's communication is small. 
This contradicts a lower bound on the communication complexity of this game. More precisely, we prove the following lower bound.
%
%An ingredient in the proof of \Cref{thm:uMv-general} is a lower bound on the communication complexity of a specific one-way communication game.
%The lower bound is given by
\begin{theorem}
	\label{thm:uMv-oneway}
	Suppose that Alice gets a uniformly random matrix $M\in \F_2^{\sqrt{n} \times \sqrt{n}}$ as input and Bob receives a uniform pair $(u,v) \in \F_2^{\sqrt{n}} \times \F_2^{\sqrt{n}}$ as input. If Alice sends a deterministic message to Bob and Bob computes $\trans{u}Mv$ such that 
	\[\Pr_{M,u,v}[\text{Bob computes }\trans{u}Mv \text{ correctly}] \geq \frac{1}{2} + \frac{1}{2^{\sqrt{n}/8}},\]
	then Alice must communicate at least $n/10$ bits.
\end{theorem}
Previously, in the \emph{randomized} two-way communication setting,  Chattopadhyay, Kouck\'{y}, Loff, and Mukhopadhyay \cite{ChattopadhyayKLM18} proved a lower bound for the game given in \Cref{thm:uMv-oneway}. Their lower bound implies the lower bound in \Cref{thm:uMv-oneway} against randomized protocols. We need a lower bound against \emph{deterministic} protocols under the {\em uniform distribution} on the inputs, and we cannot use their theorem as a black-box. We provide a straightforward proof of \Cref{thm:uMv-oneway} in \Cref{sec:uMv-oneway} by using the \emph{discrepancy method} on a related communication game (resembling a direct sum, where Bob receives multiple inputs).
% of a related communication game.

\paragraph{Preliminaries.} Before presenting the proof of \Cref{thm:uMv-general}, we define some notation. For a real valued function $f$ with a finite domain $X\times Y$, $\Ex{x,y}{f(x,y)} = \frac{1}{|X|\cdot |Y|} \cdot \sum_{x \in X, y\in Y} f(x,y)$. Similarly, for $X'\subseteq X$, $\Ex{x,y}{f(x,y) \relmiddle| x\in X'} = \frac{1}{|X'|\cdot |Y|} \cdot \sum_{x\in X', y\in Y} f(x,y)$. 
An argument in the proof of \Cref{thm:uMv-general} requires an upper bound on the number of bits to encode the contents and locations of a subset of the cells, which is given by the following proposition.
\begin{proposition}
	\label{prop:encoding-cells}
	Let $S$ be a subset of the cells of a data structure with word length $w$ and size $s$. Then, the contents and locations of $S$ can be encoded in $|S|\cdot w + |S|\cdot \log \frac{es}{|S|}$ bits.
\end{proposition}
\begin{proof}
	Since each cell stores $w$ bits, the number of bits to encode the contents is $|S|\cdot w$. Since the total number of cells is $s$, the locations can be encoded in $\log \binom{s}{|S|} \leq |S|\cdot \log \frac{es}{|S|}$ bits, where the inequality followed from \Cref{prop:binom}.
\end{proof}

\subsection{Proof of \Cref{thm:uMv-oneway}}
\label{sec:uMv-oneway}

We start by discussing a slightly related problem, whose solution will lead to the proof strategy used here. Let $M\in \F_2^{\sqrt{n}\times \sqrt{n}}$, $v\in \F_2^{\sqrt{n}}$, and $e_1,\ldots,e_{\sqrt{n}}$ be the standard basis vectors in $\F_2^{\sqrt{n}}$. Consider the communication game in which Alice gets as input a uniform random $M$ and Bob gets as input a uniform random pair $(e_i,v)$.
Bob's goal is to compute $\trans{e_i}Mv$ after receiving a message from Alice. To understand how much Alice has to communicate, it is natural to look at the problem where Bob computes $\sum_{i=1}^{\sqrt{n}} \trans{e_i}Mv_i$, where $v_1,\ldots,v_{\sqrt{n}} \in F_2^{\sqrt{n}}$. Now observe that this sum is the same as the trace of $\left(\sum_{i=1}^{\sqrt{n}}e_i\trans{v}\right)M$, which in turn is the inner product between two $n$-bit vectors. 
%Since the communication is one-way, the original protocol also solves the problem of computing $\sum_{i=1}^{\sqrt{n}} \trans{e_i}Mv_i$.  
The communication complexity of the inner product between two $n$-bit vectors is very well understood.
%, and it is a canonical example to demonstrate the \emph{discrepancy method} in communication complexity. 
Therefore, the lower bound on the amount of communication to compute the inner product between two $n$-bit vectors translates to a lower bound to the problem of computing $\trans{e_i}Mv$. This strategy applied to our setting gives us the following lower bound, which will be used to prove \Cref{thm:uMv-oneway}. Our presentation closely follows \cite[Chapter 5]{RaoY19}.
%We first present a lemma about the communication complexity of a game associated with making multiple queries to the data structure.
\begin{lemma}
\label{lem:uMv-multiplequery}
Let $0 < \epsilon \leq 1/2$ and let $k$ be an integer.
%Suppose that 
Alice gets a uniformly random $M\in \F_2^{\sqrt{n}\times \sqrt{n}}$ as input and Bob receives $k$ uniform pairs $\left(u_1,v_1\right),\ldots,\left(u_{k},v_{k}\right)\in \F_2^{\sqrt{n}} \times \F_2^{\sqrt{n}}$ as input. Assume that Alice communicates a deterministic message to Bob, and Bob computes 
$\sum_{i=1}^{k} \trans{u_i}Mv_i$ with  
\begin{align*}
\Pr_{M,u_1,v_1,\ldots,u_{k},v_{k}}\left[\text{Bob computes }\sum_{i=1}^{k} \trans{u_i}Mv_i \text{ correctly}\right] \geq  \frac{1}{2}+\epsilon.
\end{align*}
If $k\leq \sqrt{n}$, then Alice must communicate at least $9k\sqrt{n}/40 - \log (1/\epsilon)$ bits.
\end{lemma}
\begin{proof}
We use the discrepancy method to prove the communication lower bound. This requires upper bounding the discrepancy of the  \emph{communication matrix} under a given distribution. Let $R$ be a rectangle of the communication matrix, which is defined by indicator functions $A_R$ and $B_R$ such that $\left(M,\left(\left(u_1,v_1\right),\ldots,\left(u_{k},v_{k}\right)\right)\right)$ is in the rectangle $R$ if and only if $A_R(M) = 1$ and $B_R\left(\left(u_1,v_1\right),\ldots,\left(u_{k},v_{k}\right)\right) = 1$. 

Consider the distribution in which $M,\left(u_1,v_1\right),\ldots,\left(u_{k},v_{k}\right)$ are chosen at random uniformly and independently. We upper bound the discrepancy under this distribution. In other words, we claim that for every rectangle $R$,
\begin{align}
\Ex{M,\left(u_1,v_1\right),\ldots,\left(u_{k},v_{k}\right)}{A_R(M)B_R\left(\left(u_1,v_1\right),\ldots,\left(u_{k},v_{k}\right)\right) (-1)^{\sum_{i=1}^{k} \trans{u_i}Mv_i }} \leq 2\cdot 2^{-9k\sqrt{n}/40}.
\label{eqn:discrepancy}
\end{align}
By a standard relation in communication complexity between the number of bits communicated and discrepancy of rectangles (see \cite[Chapter 5, Theorem 5.2]{RaoY19}), \Cref{eqn:discrepancy} implies that Alice must communicate at least $9k\sqrt{n}/40 - \log(1/\epsilon)$ bits.
We are left with the proof of \Cref{eqn:discrepancy}.
\begin{align*}
&\left(\Ex{\left(u_1,v_1\right),\ldots,\left(u_{k},v_{k}\right)}{B_R\left(\left(u_1,v_1\right),\ldots,\left(u_{k},v_{k}\right)\right) \Ex{M}{A_R(M)(-1)^{\sum_{i=1}^{k} \trans{u_i}Mv_i }}}\right)^2 \\ 
&\leq \Ex{\left(u_1,v_1\right),\ldots,\left(u_{k},v_{k}\right)}{ B_R\left(\left(u_1,v_1\right),\ldots,\left(u_{k},v_{k}\right)\right)^2\left(\Ex{M}{A_R(M) (-1)^{\sum_{i=1}^{k} \trans{u_i}Mv_i }}\right)^2} \\
&\leq \Ex{\left(u_1,v_1\right),\ldots,\left(u_{k},v_{k}\right)}{\left(\Ex{M}{A_R(M) (-1)^{\sum_{i=1}^{k} \trans{u_i}Mv_i }}\right)^2}.
\end{align*}
where the first inequality follows from convexity and the second one follows from the fact that $B_R\left(\left(u_1,v_1\right),\ldots,\left(u_{k},v_{k}\right)\right) \leq 1$.
Now 
\begin{align*}
&\Ex{\left(u_1,v_1\right),\ldots,\left(u_{k},v_{k}\right)}{\left(\Ex{M}{A_R(M) (-1)^{\sum_{i=1}^{k} \trans{u_i}Mv_i }}\right)^2}\\
&\leq \Ex{\left(u_1,v_1\right),\ldots,\left(u_{k},v_{k}\right), M,M'}{A_R(M)A_R(M') (-1)^{\sum_{i=1}^{k} \trans{u_i}Mv_i +\sum_{i=1}^{k} \trans{u_i}M'v_i }} \\
&\leq \Ex{M,M'}{\left|\Ex{\left(u_1,v_1\right),\ldots,\left(u_{k},v_{k}\right)}{(-1)^{\sum_{i=1}^{k} \trans{u_i}(M+M')v_i} }\right|} \\
&= \Ex{M}{\left|\Ex{(u,v)}{(-1)^{\trans{u}Mv}}\right|^{k}},
\end{align*}
where the last equality follows from the fact that $(u_1,v_1),\ldots,(u_k,v_k)$ are chosen independent of each other and $M+M'$ is uniformly distributed as $M$ and $M'$ are chosen uniformly and independently at random.
We are left with upper bounding $\Ex{M}{\left|\Ex{(u,v)}{(-1)^{\trans{u}Mv}}\right|^{k}}$. First note that if $M$ has rank $r$, then $\Ex{u,v} {(-1)^{\trans{u}Mv}} = 2^{-r}$. This is because,
	\begin{align*}
	\Ex{u,v} {(-1)^{\trans{u}Mv}} = \frac{1}{2^{\sqrt{n}}} \cdot \left(\sum_{v:Mv=0} 1 \right)+\frac{1}{2^{\sqrt{n}}} \cdot \left(\sum_{v:Mv\neq 0} \Ex{u}{(-1)^{\trans{u}Mv}} \right) 
	= \frac{2^{\sqrt{n}-r}}{2^{\sqrt{n}}} + 0 = 2^{-r}.
	\end{align*}
	In addition, $\Pr_M\left[\text{rank of }M \leq 9\sqrt{n}/20\right]\leq 2^{-9n/10}$. Indeed, the number of matrices in $\F_2^{\sqrt{n}\times \sqrt{n}}$ of rank at most $k$ is at most 
	\[\binom{2^{\sqrt{n}}}{k} \cdot \left(2^{k}\right)^{\sqrt{n}} \leq 2^{2k\sqrt{n}}.\]	
	Therefore, using the law of total expectation, we have that
	\begin{align*}
	\Ex{M}{\left|\Ex{(u,v)}{(-1)^{\trans{u}Mv}}\right|^{k}}\leq \Pr_M\left[\text{rank of }M \leq 9\sqrt{n}/20\right] + 2^{-9k\sqrt{n}/20} \leq 2\cdot 2^{-9k\sqrt{n}/20},
	\end{align*}
	where the last inequality followed from the fact that $k \leq \sqrt{n}$.
\end{proof}

\begin{proof}[Proof of \Cref{thm:uMv-oneway}] Let $c$ be the number of bits communicated by Alice. We show that $c > n/10$.
Define $Z_M(u,v) = 1$ if Bob correctly computes $\trans{u}Mv$ and $Z_M(u,v) = -1$ otherwise. 
By the definition of $Z_M(u,v)$ and the lower bound on the probability of Bob's computation being correct, we have that $\Ex{M,u,v}{Z_M(u,v)} \geq 2\cdot 2^{-\sqrt{n}/8}$.

We note that it is without loss of generality that $\Ex{u,v}{Z_M(u,v)}\geq 0$ for every $M\in \F_2^{\sqrt{n}\times \sqrt{n}}$. This is because Alice on input $M$ can send an extra bit indicating whether $\Ex{u,v}{Z_M(u,v)}<0$ and Bob will flip his output accordingly. 
%Hence, it is sufficient to show that $c > n/10$.

We now use the given protocol to design a protocol for a new communication game: Suppose that Alice gets a uniformly random $M\in \F_2^{\sqrt{n}\times \sqrt{n}}$ as input and Bob receives $\sqrt{n}$ uniform pairs $(u_1,v_1),\ldots,(u_{\sqrt{n}},v_{\sqrt{n}})\in \F_2^{\sqrt{n}} \times \F_2^{\sqrt{n}}$ as input. 
We will use \Cref{lem:uMv-multiplequery} with $k=\sqrt{n}$ to obtain the  desired lower bound on $c$.

We claim that there is a communication protocol in which Alice communicates $c$ bits and Bob computes $\sum_{i=1}^{\sqrt{n}} \trans{u_i}Mv_i$ such that 
\begin{align}
\Pr_{M,u_1,v_1,\ldots,u_{\sqrt{n}},v_{\sqrt{n}}}\left[\text{Bob computes }\sum_{i=1}^{\sqrt{n}} \trans{u_i}Mv_i \text{ correctly}\right] \geq  \frac{1}{2}+ \frac{2^{\sqrt{n}-1}}{2^{n/8}}.
\label{eqn:prob_amp}
\end{align}
Alice's message is same as before, and Bob computes each of $\trans{u_i}Mv_i$ separately and outputs the sum modulo $2$. We now prove \Cref{eqn:prob_amp}.
For a fixed $M$, the probability that Bob correctly computes $\sum_{i=1}^{\sqrt{n}} \trans{u_i}Mv_i $ is $\frac{1}{2}\left(1 + \left(\Ex{u,v}{Z_M(u,v)}\right)^{\sqrt{n}}\right)$. Therefore the overall probability that Bob correctly computes $\sum_{i=1}^{\sqrt{n}} \trans{u_i}Mv_i $ is at least 
\[\frac{1}{2}\left(1 + \frac{\sum_{M} (\Ex{u,v}{Z_M(u,v)})^{\sqrt{n}}}{2^n}\right) \geq \frac{1}{2}\left(1 + \left(\Ex{M,u,v}{Z_M(u,v)}\right)^{\sqrt{n}}\right) \geq \frac{1}{2}+ \frac{2^{\sqrt{n}-1}}{2^{n/8}}, \]
where the first inequality follows from convexity of the function $f(x) = x^k$ with $k=\sqrt{n}$.
Applying \Cref{lem:uMv-multiplequery} with $k=\sqrt{n}$ implies that $c > n/10$, which completes the proof of the theorem.
\end{proof}

\subsection{Proof of \Cref{thm:uMv-general}}
\label{sec:uMv-general}
%TODO: fix $\alpha = 2(w+\log (sw/n))$ in the theorem statement.

If $n< 36$, the theorem is vacuously true as $c \leq 1/36$. For the rest of the argument we will assume that $n \geq 36$.
We prove a lower bound on the query time $t$ against deterministic data structures with space $s$ and word size $w$. Suppose that the input data $M$ and query $u\trans{v}$ is given by choosing $M,u,v$ uniformly and independently at random, and the query algorithm is guaranteed to satisfy
\[\Pr_{M,u,v}\left[\text{query algorithm computes }\trans{u}Mv \text{ correctly}\right] \geq \frac{1}{2} + 2^{-\sqrt{n}/16}.\] 
By Yao's minmax principle, this will imply a lower bound on randomized data structures. 

We first modify the given data structure to ensure that for every $M\in \F_2^{\sqrt{n}\times \sqrt{n}}$, the probability that $\trans{u}Mv$ is correctly computed is at least $1/2$.
Assume that we have a data structure with query time $t'$, space $s'$ and word size $w$. The modified data structure stores an extra bit indicating whether the
$\Pr_{u,v}\left[\text{query algorithm computes }\trans{u}Mv \text{ correctly}\right]$ is less than $1/2$ or not for a given $M$. The query algorithm is the same as before, but accesses this extra bit to flip the output if it is set to $1$. Clearly, the new data structure has query time $t = t'+1$, space $s = s'+1$ and word size $w$. Moreover, under this modification, we have
\begin{itemize}
\item $\Pr_{M,u,v}\left[\text{query algorithm computes }\trans{u}Mv \text{ correctly}\right] \geq 1/2 + 2^{-\sqrt{n}/16}$.
\item $\Pr_{u,v}\left[\text{query algorithm computes }\trans{u}Mv \text{ correctly}\right] \geq 1/2$ for every $M$.
\end{itemize}
In the rest of the proof, we work with this modification and show that  $t \geq \Omega\left(\min\left\{ \frac{n}{\beta}, \frac{\sqrt{n}}{\log \frac{s\beta}{n}}\right\}\right)$, where $\beta = 2(w+\log sw/n)$. Observe that $\beta \leq n/256$; otherwise the lower bound is vacuous.

Assume by contradiction that $t\leq \min\left\{ \frac{n}{256\beta}, \frac{\sqrt{n}}{256 \log \frac{s\beta}{n}}\right\}$. 
Define $Z_M(u,v)=1$ if the query algorithm correctly computes $\trans{u}Mv$, and $-1$ otherwise. We have
\begin{align}
\Ex{M,u,v}{Z_M(u,v)} = 2\cdot \Pr_{M,u,v}\left[\text{query algorithm computes }\trans{u}Mv \text{ correctly}\right] - 1 \geq 2\cdot 2^{-\sqrt{n}/16}. \label{eqn:ZMguarantee}
\end{align}
Note that $\Ex{M,u,v}{Z_M(u,v)}$ captures the \emph{advantage} or \emph{bias} of the data structure - it is much more convenient to work with the advantage than the probability of the query algorithm being correct.

% For an $M$, let $Q_1$ denote the set of queries on which the query algorithm correctly computes $\trans{u}Mv$, and let $Q_2$ denote the set of queries on which the query algorithm incorrectly computes $\trans{u}Mv$. 
%By definition, we know that 
%\[\Pr_{u,v}\left[u\trans{v}\in Q_1\right] - \Pr_{u,v}\left[u\trans{v}\in Q_2\right] = \Ex{u,v}{Z_M(u,v)} \geq 0.\] 
The following lemma, a variant of cell sampling, guarantees the existence of a small subset $S$ of cells such that a large number of queries $Q'$ can be computed by only accessing $S$, and $\Ex{u,v}{Z_M(u,v) \mid u\trans{v}\in Q'} \approx \Ex{u,v}{Z_M(u,v)}$.
%the difference between the fraction of queries in $Q'$ correctly computed and incorrectly computed is proportional to $\Ex{u,v}{Z_M(u,v)}$. We defer the proof of this lemma to later.
%which guarantees that if the query time is small and the fraction of queries on which the query algorithm is correct is much more than the fraction of queries on which the query algorithm is incorrect, then there is a small subset of cells such that  among the queries answered by this subset of cells, the fraction computed correctly is much more than the fraction computed incorrectly.
\begin{lemma}
	\label{lem:cellsampling}
	Let $M\in \F_2^{\sqrt{n}\times \sqrt{n}}$. Define $Q_1 = \left\{u\trans{v}\relmiddle| Z_M(u,v) = 1\right\}$ and $Q_2 = \left\{u\trans{v}\relmiddle| Z_M(u,v) = -1\right\}$.
	%, and let $c_1>1$ be a constant and $\alpha = w + \log \frac{sw}{n}$. 
	If $t\leq \min\left\{ \frac{n}{256\beta}, \frac{\sqrt{n}}{256 \log \frac{s\beta}{n}}\right\}$, then there exits a subset of cells $S$, and subsets $Q_1'\subseteq Q_1$ and $Q_2' \subseteq Q_2$ such that 
	\begin{enumerate}
		\item $|S| = \ceil{\frac{n}{128\beta}}$,
		\item $\Pr_{u,v}\left[u\trans{v}\in Q'_1\right] - \Pr_{u,v}\left[u\trans{v}\in Q'_2\right] \geq \Ex{u,v}{Z_M(u,v)} \cdot 2^{-\sqrt{n}/16}$, 
		\item $Q_1' \cup Q_2'$ is the set of all queries computed by accessing cells only in $S$.
	\end{enumerate}
\end{lemma}

% \autoref{sec:cellsampling-proofs}.

\begin{figure}[t]
\begin{algorithm}[H]
\KwIn{Alice's input is $M$ and Bob's input is $(u,v)$}
\KwOut{Alice communicates a deterministic message and Bob computes $\trans{u}Mv$.}
\vspace{2ex}
Let $Q_1 = \left\{u\trans{v}\relmiddle| Z_M(u,v)=1\right\}$ and  $Q_2 =\left\{u\trans{v}\relmiddle| Z_M(u,v)=-1\right\}$\;
Apply \Cref{lem:cellsampling} with $Q_1, Q_2$ to obtain a subset of cells $S$ and subsets $Q_1' \subseteq Q_1$ and $Q_2' \subseteq Q_2$\;
%such that all queries in $Q = Q_1' \cup Q_2'$ are answered by accessing cells only in $S$.
Let $b\in \{0,1\}$ be such that $\Pr_{u,v}\left[\trans{u}Mv=b\mid u\trans{v}\notin Q'\right] \geq \Pr_{u,v}\left[\trans{u}Mv=1-b\mid u\trans{v}\notin Q'\right]$, where $Q' = Q'_1 \cup Q'_2$\;
Alice communicates $b$ followed by locations and contents of $S$\;
\lIf{$u\trans{v} \in Q'$}{
Bob uses the query algorithm to compute $\trans{u}Mv$}
\lElse{Bob outputs $b$}
\NoCaptionOfAlgo
\end{algorithm}
\caption{One-way protocol on inputs $M,(u,v)$ computing $\trans{u}Mv$.}
\label{prot:ds-simulation}
\end{figure}

We move on to the final step of the proof of \Cref{thm:uMv-general}. What is left is to design a one-way protocol using the sets guaranteed by \Cref{lem:cellsampling}.
The protocol is described in  \Cref{prot:ds-simulation}. We will show the validity of this protocol by showing that both Alice and Bob know the subset $Q'$ of queries. Since Alice's input is $M$, she knows the contents of all the cells, which gives $S$. With regard to knowing $Q'$, the locations and contents of cells in $S$ suffice. This is because the query algorithm can be simulated on all queries to check if any cell outside of $S$ is being accessed. We are proving \Cref{thm:uMv-general} by contradicting \Cref{thm:uMv-oneway}, which
%The contradiction to \Cref{thm:uMv-oneway} 
is achieved by the following.

\begin{lemma}
\label{lem:prot-guarantee}
The protocol in \Cref{prot:ds-simulation} has the following guarantees
%\begin{enumerate}
%\item[(a)] Alice communicates fewer than $n/10$ bits.
%\item[(b)] $\Pr_{M,u,v}\left[\text{Bob computes }\trans{u}Mv \text{ correctly}\right] > 1/2 + 2/2^{-\sqrt{n}/8}.$
%\end{enumerate}
(a) Alice communicates fewer than $n/10$ bits, and
(b) $\Pr_{M,u,v}\left[\text{Bob computes }\trans{u}Mv \text{ correctly}\right] \geq 1/2 + 1/2^{\sqrt{n}/8}.$
\end{lemma}
%Clearly, \Cref{lem:prot-guarantee} contradicts \Cref{thm:uMv-oneway}.

Now, we need to prove \Cref{lem:cellsampling,lem:prot-guarantee} to complete the proof of \Cref{thm:uMv-general}.

\begin{proof}[Proof of \Cref{lem:cellsampling}]
	Let $S$ be a uniformly random subset of the cells of size $|S| = \ceil{\frac{n}{128\beta}}$. Define $D(u,v,S) = Z_M(u,v)$ if the query algorithm only accesses cells in $S$ to compute $\trans{u}Mv$; otherwise $D(u,v,S) = 0$.
	By linearity of expectation,
	\begin{align*}
	\Ex{u,v,S}{D(u,v,S)} = \Ex{u,v}{Z_M(u,v)} \cdot \frac{\binom{s-t}{|S|-t}}{\binom{s}{|S|}} &=\Ex{u,v}{Z_M(u,v)} \cdot \frac{|S| \cdot (|S|-1)\cdots (|S|-t+1)}{s \cdot (s-1)\cdots (s-t+1)} \\
	&\geq \Ex{u,v}{Z_M(u,v)} \cdot \left(\frac{|S|-t}{s}\right)^{t}.
	\end{align*}
	Recall that $|S| \geq \frac{n}{128\beta}$ and $t \leq \frac{n}{256\beta}$. Moreover, $\beta = 2(w+\log sw/n) \geq 2$. This implies that 
	\[ \left(\frac{|S|-t}{s}\right)^{t} \geq 2^{-t\cdot \log \frac{256s\beta}{n}} \geq 2^{-16t\cdot \log \frac{s\beta}{n}}.\] 
	So we get $\Ex{u,v,S}{D(u,v,S)} \geq \Ex{u,v}{Z_M(u,v)} \cdot 2^{- 16 \cdot t\cdot \log\frac{s\beta}{n}}$.
	Therefore, there exists an $S$ such that 
	\[\Ex{u,v}{D(u,v,S)} \geq \Ex{u,v}{Z_M(u,v)} \cdot 2^{- 16 \cdot t\cdot \log\frac{s\beta}{n}} \geq \Ex{u,v}{Z_M(u,v)} \cdot 2^{-\sqrt{n}/16},\]
	where the last inequality follows from the fact that $16 \cdot t\cdot \log\frac{s\beta}{n} \leq  \sqrt{n}/16$.  Setting 
	\[Q_1' = \left\{u\trans{v}\in Q_1 \mid D(u,v,S) = 1\right\} \text{ and }  Q_2'=\left\{u\trans{v}\in Q_2\mid D(u,v,S) = -1\right\}\] 
	completes the proof of the lemma.
\end{proof}

\begin{proof}[Proof of  \Cref{lem:prot-guarantee}]
We first prove part (a). Recall that $\beta = 2\left(w + \log \frac{sw}{n}\right)$.
Let $c$ be the number of bits communicated by Alice. By \Cref{prop:encoding-cells} and the definition of $\beta$,
\begin{align*}
c &\leq 1 + \ceil{\frac{n}{128\beta}}\cdot w + \ceil{\frac{n}{128\beta}} \cdot \log \frac{128e \cdot s\beta}{n} \\
&= 1 +  \ceil{\frac{n}{128\beta}} \cdot \left(w+\log \frac{s\beta}{n}\right) +\ceil{\frac{n}{128\beta}} \cdot \log 128e.
\end{align*}
Since $\beta \geq 2w$, $\beta \geq 2\log \frac{s}{n}$ and $\beta \geq \log \beta$, we get that $w+\log \frac{s\beta}{n} \leq 2\beta$. Moreover, using the fact that $\ceil{\frac{n}{128\beta}} \leq \frac{n}{128\beta}+1$, $\beta \geq 2$ and $\beta \leq n/256$, we can say that 
\begin{align*}
c &\leq  1 + \frac{2n}{128} +2\beta+ \frac{n \log 128e}{128\beta} + \log 128e \\
&\leq 1 + \frac{2n}{128} + \frac{4.5n}{128(\beta/2)} + \frac{n}{128} + \log 128e\leq 10 + \frac{7.5n}{128} < \frac{n}{10},
\end{align*}
where the last inequality follows from $n \geq 36$.

We now prove part (b) of the claim. 
Define $Z'_M(u,v)=1$ if the Bob correctly computes $\trans{u}Mv$ and $Z'_M(u,v)=-1$ otherwise. The probability with which Bob correctly computes $\trans{u}Mv$ is given by $\left(1+\Ex{M,u,v}{Z'_M(u,v)}\right)/2$. We will show that $\Ex{M,u,v}{Z'_M(u,v)} \geq 2 \cdot 2^{-\sqrt{n}/8}$, which will imply that the probability of being correct is at least $1/2 + 2^{-\sqrt{n}/8}$. 

Let $Q_1, Q_2, Q'_1, Q'_2,$ and $Q'$ be as defined in the protocol in \Cref{prot:ds-simulation}. We first establish some properties about these sets. We know that $\Pr_{u,v}[u\trans{v}\in Q_1] - \Pr_{u,v}[u\trans{v}\in Q_2]  = \Ex{u,v}{Z_M(u,v)}$. Moreover, the application of  \Cref{lem:cellsampling} in the protocol is valid since  $t\leq \frac{n}{256\alpha}$, and hence 
%Observe that
\begin{align}
\Pr_{u,v}\left[u\trans{v}\in Q'_1\right] - \Pr_{u,v}\left[u\trans{v}\in Q'_2\right] \geq \Ex{u,v}{Z_M(u,v)} \cdot 2^{-\sqrt{n}/16}. \label{eqn:fraction-adv}
\end{align}
Since Bob can simulate the query algorithm on $Q'$ by accessing only $S$, which is guaranteed by \Cref{lem:cellsampling}, we have
\begin{align*}
\Ex{u,v}{Z'_M(u,v)} &= \Pr_{u,v}\left[u\trans{v}\in Q'\right]\cdot \left(\Pr_{u,v}\left[u\trans{v}\in Q_1' \mid u\trans{v}\in Q'\right] -\Pr_{u,v}\left[u\trans{v}\in Q_2' \mid u\trans{v}\in Q'\right]\right) \\
&\quad  + \Pr_{u,v}\left[u\trans{v} \notin Q'\right] \left(\Pr_{u,v}\left[\trans{u}Mv=b \mid u\trans{v}\notin Q'\right]-\Pr_{u,v}\left[\trans{u}Mv=1-b \mid u\trans{v}\notin Q'\right]\right) \\
&\geq \left(\Pr_{u,v}\left[u\trans{v}\in Q_1'\right] -\Pr_{u,v}\left[u\trans{v}\in Q_2'\right] \right) \geq \Ex{u,v}{Z_M(u,v)} \cdot 2^{-\sqrt{n}/16},
\end{align*}
where the first inequality follows from the choice of $b$ and the second inequality used \Cref{eqn:fraction-adv}. 

To conclude, 
\begin{align*}
\Ex{M,u,v}{Z'_M(u,v)} = \Ex{M}{\Ex{u,v}{Z'_M(u,v)}} &\geq \Ex{M}{\Ex{u,v}{Z_M(u,v)} \cdot 2^{-\sqrt{n}/16}} \\
&= \Ex{M}{\Ex{u,v}{Z_M(u,v)}} \cdot 2^{-\sqrt{n}/16} \\
&= \Ex{M,u,v}{Z_M(u,v)} \cdot 2^{-\sqrt{n}/16} \geq 2\cdot 2^{-\sqrt{n}/8},
\end{align*}
where the last inequality follows from
\Cref{eqn:ZMguarantee}.
% implies that $ \Ex{M,u,v}{Z_M(u,v)} \geq 2 \cdot 2^{-\sqrt{n}/16}$, and hence $\Ex{M,u,v}{Z'_M(u,v)} \geq 2\cdot 2^{-\sqrt{n}/8}$.
\end{proof}

%
%\section{Conclusion}
%
\subsection*{Acknowledgments}
We thank Paul Beame, Sajin Koroth, Pavel Hrube\v{s}, Pavel Pudl{\' a}k, Anup Rao, Makrand Sinha, Amir Yehudayoff and Sergey Yekhanin for useful discussions. Special thanks to Paul, Anup, Makrand and Amir for the encouragement to write up these results.
%S.N.R. is supported by the National Science Foundation under agreement CCF-1420268. C.R. is supported by the .... Part of this work was done while C.R. was a graduate student at the University of Washington, Seattle.

\bibliographystyle{plain}
\bibliography{ref}
%\newpage
%\input{appendix}
\end{document}